\documentclass[superscriptaddress,a4paper, amsfonts, amssymb, amsmath, reprint, showkeys, nofootinbib, twoside,notitlepage,onecolumn]{revtex4-1}

\def\prd{{Phys.~Rev.~D}}        
\usepackage{amsmath,amstext}
\usepackage[T1]{fontenc}
\usepackage{amssymb}
\usepackage{graphicx}
\usepackage{ae,aecompl}

\DeclareFontFamily{OT1}{pzc}{}
\DeclareFontShape{OT1}{pzc}{m}{it}{<-> s * [1.10] pzcmi7t}{}
\DeclareMathAlphabet{\mathpzc}{OT1}{pzc}{m}{it}

\usepackage{hyperref}
\usepackage{amsmath}
\usepackage{amssymb}
\usepackage{mathtools}
\usepackage{bm}
\usepackage{cleveref}
\usepackage{tensor}
\usepackage{braket}
\usepackage{enumitem}
\usepackage{mhchem}
\usepackage{amsthm}
\usepackage{nccmath}
\usepackage{mathrsfs}
\usepackage{color}

\newcommand{\spc}{\quad \quad \quad}

\def\be{\begin{equation}}
\def\ee{\end{equation}}
\def\beq{\begin{eqnarray}}
\def\eeq{\end{eqnarray}}

\theoremstyle{definition}

\theoremstyle{theorem}
\newtheorem{theorem}{Theorem}

\begin{document}
\title{The initial data of effective field theories of relativistic viscous fluids and gravity}
\author{Lorenzo Gavassino} 
\affiliation{Department of Applied Mathematics and Theoretical Physics, University of Cambridge, Wilberforce Road, Cambridge CB3 0WA, United Kingdom}
\author{\'Aron D. Kov\'acs}
\affiliation{Centre for Geometry, Analysis and Gravitation, School of Mathematical Sciences, Queen Mary University of London, Mile End Road, London E1 4NS, United Kingdom}
\author{Harvey S. Reall}
\affiliation{Department of Applied Mathematics and Theoretical Physics, University of Cambridge, Wilberforce Road, Cambridge CB3 0WA, United Kingdom}

\begin{abstract}
There has been recent progress in developing well-posed theories of relativistic viscous hydrodynamics and of gravitational effective field theories. These have in common the feature that they introduce unphysical degrees of freedom. We address the problem of how these should be treated. We propose a ``reduction of order'' approach which is applied not at the level of equations of motion but only to initial data. This specifies uniquely the data for the unphysical modes in terms of the data for the physical modes. We argue that the apparent breaking of Lorentz invariance associated with this approach is not a problem provided one restricts to Lorentz frames for which the assumptions of effective field theory are manifestly valid.
\end{abstract} 
\maketitle

\section{Introduction}
\vspace{-0.4cm}

In this paper, we will address a problem that occurs both in theories of relativistic viscous hydrodynamics and in effective field theory (EFT) extensions of General Relativity. In hydrodynamics, a recent breakthrough is the discovery by Bemfica, Disconzi, Noronha, and Kovtun (BDNK) of equations of motion for a relativistic viscous fluid for which the initial value problem is well-posed \cite{Bemfica2019_conformal1,Kovtun2019,BemficaDNDefinitivo2020}. A nice way of understanding this is to adopt an EFT perspective. Starting in, say, the Landau frame \cite{landau6}, EFT field redefinitions can be used to bring the equations of motion to the BDNK form \cite{Kovtun2019}. Similarly, in gravitational physics, for a very general class of EFTs it has been shown by Figueras, Held and Kov\'acs (FHK) \cite{Figueras:2024bba} that field redefinitions can be used to obtain a theory whose equations of motion admit a well-posed initial value problem, at least within the regime of validity of EFT.

A feature common to both of these developments is the introduction of unphysical degrees of freedom \cite{Kovtun2019,GavassinoLyapunov_2020,MullinsGavassinoFluctuwithInfo2023ott,GavassinoFirstOrderWithInfo2024ufs,HellerOrderByOrder:2025dxh}. In each case we have an EFT which is supposed to describe certain (``light'' or ``slow'') degrees of freedom but the well-posed systems of equations just described contain additional (``heavy'' or ``fast") degrees of freedom. In the BDNK theory, this occurs because the equations of motion are second-order in time derivatives, as opposed to the Landau-frame equation, which are of first order in time in the fluid's local rest frame. Hence, in the initial-value problem, one must specify additional initial data, indicating the presence of unphysical degrees of freedom. In the case of a gravitational EFT, the equations involve higher than second derivatives of the metric tensor and so, even though the FHK formulation is well-posed, one still needs to specify initial data for these unphysical massive degrees of freedom. 

Is the presence of these unphysical degrees of freedom a problem? In the case of hydrodynamics, they should rapidly decay by dissipation, leaving an asymptotic solution in which only the light degrees of freedom are excited \cite{Geroch95,LindblomRelaxation1996}. But how do we know that this is the (physically) {\it correct} solution arising from the prescribed initial data for these light degrees of freedom? In gravity, there is no dissipation so the problem is worse: the additional degrees of freedom exhibit high-frequency oscillations (or rapid exponential growth), which can infect the physical degrees of freedom.

We will discuss a simple way of dealing with this problem. The idea is based on the ``reduction of order'' approach to the equations of an EFT, in which one repeatedly substitutes lower order equations into higher derivative terms in order the reduce the order of the highest time derivatives to $1$ (hydrodynamics \cite{Armas:2020mpr,Basar:2024qxd,BhambureDensityFrame2024axa,GavassinoParabolic2025hwz}) or $2$ (gravity \cite{Parker:1993dk,Flanagan:1996gw}). This procedure has mostly been discussed as a way of eliminating the unphysical degrees of freedom at the level of the equations of motion. But this procedure breaks Lorentz invariance (or general covariance) and often the resulting equations do {\it not} admit a well-posed initial value problem \cite{Flanagan:1996gw}, so this may not always be a useful approach for e.g. performing numerical simulations. Instead, we will retain here the covariant and well-posed higher-order equations of motion, and we will use the reduction of order procedure only to fix the initial data for higher time derivatives of the fields. This corresponds to fixing uniquely the initial data for the unphysical degrees of freedom in terms of the data for the physical degrees of freedom. For gravity, this approach was suggested by FHK \cite{Figueras:2024bba} and has been used for scalar field EFTs in \cite{Figueras:2025wtx}. However, for hydrodynamics it does not appear to have been discussed or implemented previously (e.g., it is not used in the numerical simulations of \cite{Pandya:2021ief,Pandya:2022pif,Pandya:2022sff,Clarisse:2025lli}).

A concern about this approach is that the reduction of order requires a choice of time coordinate, and therefore still breaks Lorentz or diffeomorphism invariance, although now at the level of initial data instead of the equations of motion. So one might be concerned that results obtained using this method will depend on the choice of time function. However, we will argue that this is not the case, provided one restricts to Lorentz frames (or coordinate charts) for which EFT is manifestly valid, in the sense that gradients of fields are small compared to the UV scale.

Throughout the article, we adopt the metric signature $(-,+,+,+)$, and work in natural units, $c=\hbar=k_B=1$.

\newpage
\vspace{-0.3cm}
\section{Two toy models}
\vspace{-0.2cm}

To warm up, we consider here a couple of causal and stable linear effective field theories of matter in 1+1 dimensions (for this kind of theory, stability automatically implies well-posedness forward in time \cite[\S 3.10]{rauch2012partial}). These theories possess extra degrees of freedom, which need to be initialized properly. We will demonstrate that, within the EFT framework, there is only one correct way to do so. 

\vspace{-0.3cm}
\subsection{Causal heat conduction in a moving rod}\label{sec:rod}
\vspace{-0.2cm}

We consider an infinitely long metallic rod aligned with the $x$-axis and moving at a constant speed $v$ along this direction. Our goal is to study the evolution of its temperature profile, $T(t,x)$. In the BDNK framework, a natural hyperbolic (and well-posed) EFT formulation of the advection-diffusion equation is \cite{GavassinoAntonelli:2025umq}
\begin{equation}\label{heatBDNK}
\gamma(\partial_t {+}v\partial_x)T =\lambda (\partial^2_x{-}\partial^2_t) T +\mathcal{O}(\lambda^2) \, ,
\end{equation}
where $\gamma = (1-v^2)^{-1/2}$ and the diffusion coefficient $\lambda$ is of the order of the mean free path, and may be regarded as the UV cutoff scale of the effective theory. That is, the theory is expanded in powers of $\lambda\partial$. Now, since equation \eqref{heatBDNK} is of second order in time, we need to fix some initial profiles for both $T$ and $\partial_t T$. However, let us note that, according to \eqref{heatBDNK}, $\partial_t T=-v\partial_x T+\mathcal{O}(\lambda)$, which also implies $\partial_t^2 T=v^2\partial^2_x T+\mathcal{O}(\lambda)$. Hence, we can replace the second time derivative in \eqref{heatBDNK} with a second space derivative, as the related error is of order $\lambda^2$. The result is \cite{GavassinoParabolic2025hwz}
\begin{equation}\label{dt1heatBDNK}
(\partial_t {+}v\partial_x)T =\gamma^{-3}\lambda \partial^2_x T +\mathcal{O}(\lambda^2) \, . 
\end{equation}
Now, the key observation is that, if we are in a regime where \eqref{heatBDNK} applies (that is, if terms of order $\lambda^2$ can be neglected), then we are automatically in a regime where \eqref{dt1heatBDNK} also holds. Consequently, when specifying initial data for \eqref{heatBDNK}, we have no freedom: the condition \eqref{dt1heatBDNK} must be enforced. Therefore, the only physically sensible Cauchy problem for the BDNK-type equation \eqref{heatBDNK} is
\vspace{-0.2cm}
\begin{equation}\label{correctCauchy}
\begin{cases}
(\partial_t {+}v\partial_x)T =\gamma^{-1}\lambda (\partial^2_x{-}\partial^2_t) T \, , \\
T|_{t=0}=T_0 \, , \\
\partial_t T|_{t=0}=-v\partial_x T_0 +\gamma^{-3}\lambda\partial^2_x T_0 \, ,\\ 
\end{cases}
\end{equation}
which allows us to freely specify only the initial temperature profile $T_0(x)$.

It is worth emphasizing that equation \eqref{dt1heatBDNK} is not new. In fact, within the ``density-frame'' formulation of hydrodynamics \cite{Armas:2020mpr,Basar:2024qxd}, one replaces \eqref{heatBDNK} with \eqref{dt1heatBDNK} altogether, thereby eliminating the additional degrees of freedom directly at the level of the equations of motion (albeit at the cost of forfeiting causality and formal covariance, although well-posedness is preserved). Here, however, we are proposing a different perspective. If one chooses to retain \eqref{heatBDNK}, then \eqref{dt1heatBDNK} should instead be regarded as a tool for constraining the initial data, ensuring that the redundant degrees of freedom are properly fixed. Any alternative prescription for initializing $\partial_t T$ would unavoidably introduce spurious non-hydrodynamic modes into the initial state, leading to a rapid transient relaxation on a timescale set by the UV cutoff $\lambda$, as shown in the example of figure \ref{fig:heat}.

\begin{figure}[b!]
    \centering
\includegraphics[width=0.43\linewidth]{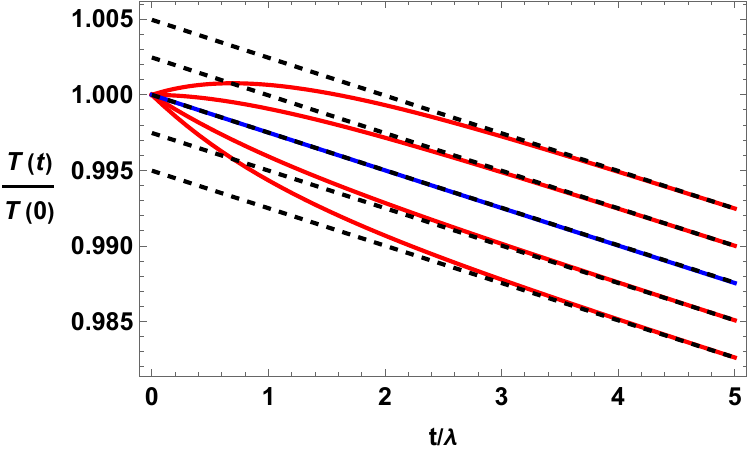}
\caption{Comparison between various solutions of \eqref{heatBDNK} (solid lines), under the ansatz $T(t,x)=T(t)e^{i k x}$, with $\lambda k=0.05$ and $v=0$. The red curves represent solutions initialized with the same $T(0)$, but various choices of $\partial_t T(0)\sim \lambda k^2$. These solutions exhibit an initial transient relaxation on a timescale $t \sim \lambda$, after which they converge to solutions of \eqref{dt1heatBDNK} (dashed) with a \textit{different} $T(0)$. The only trajectory that does not display any transient behavior is the one arising from the Cauchy problem \eqref{correctCauchy} (blue~curve), which perfectly overlaps the one solution of \eqref{dt1heatBDNK} having the same $T(0)$. Note that, in general, the order-reduced equation may not be well-posed. Here, this property happens to hold, so we can actually contrast solutions of \eqref{heatBDNK} and \eqref{dt1heatBDNK}.}
    \label{fig:heat}
\end{figure}

\subsection{Dispersive sound in ultradense matter}

As an example of a non-dissipative effective field theory, we consider the Ruderman-Bludman model of sound dispersion in ultradense matter \cite{Bludman1968,Fox1970FTL,GavassinoBounds2023myj}. Specifically, let $\phi(t,x)$ be the density perturbation of an otherwise uniform medium at rest. Then, to leading order in some microscopic lengthscale $\lambda$, we can write
\begin{equation}\label{dispersive}
(\partial_t^2{-}c_s^2 \partial_x^2)\phi +\lambda^2(\partial^2_t{-}\partial^2_x)^2 \phi+\mathcal{O}(\lambda^4) =0 \, ,
\end{equation}
where $c_s$ is the infrared speed of sound. This equation describes deviations from the usual phononic dispersion relation $\text{``frequency''}=c_s \times \text{``wavenumber''}$. It is certainly a causal theory, and it was shown in \cite{GavassinoBounds2023myj} that it is covariantly stable (and thus well-posed) if and only if $c_s^2\in [0,1]$. It is fourth-order in time, and thus possesses two extra degrees of freedom besides the two usual sound waves. To initialize these, we follow the same strategy as before. Since equation \eqref{dispersive} entails $\partial^2_t \phi=c_s^2\partial^2_x \phi+\mathcal{O}(\lambda^2)$, we can replace the time derivatives with space derivatives in the $\mathcal{O}(\lambda^2)$-term, as the price that we pay can be reabsorbed into $\mathcal{O}(\lambda^4)$. The result is
\begin{equation}\label{dispersiveacausal}
(\partial_t^2{-}c_s^2 \partial_x^2)\phi +\lambda^2(1{-}c_s^2)^2 \partial^4_x \phi+\mathcal{O}(\lambda^4) =0\, .
\end{equation}
This equation and its time derivative uniquely fix the initial values of $\partial_t^2\phi$ and $\partial^3_t \phi$. The resulting Cauchy problem is
\begin{equation}\label{cauchysound}
\begin{cases}
(\partial_t^2{-}c_s^2 \partial_x^2)\phi +\lambda^2(\partial^2_t{-}\partial^2_x)^2 \phi =0 \, , \\
\phi|_{t=0}=\phi_0 \, , \\
\partial_t\phi|_{t=0}=\phi_1 \, , \\
\partial_t^2\phi|_{t=0}=c_s^2 \partial^2_x\phi_0-\lambda^2(1{-}c_s^2)^2\partial^4_x\phi_0 \, , \\
\partial_t^3\phi|_{t=0}=c_s^2 \partial^2_x\phi_1-\lambda^2(1{-}c_s^2)^2\partial^4_x\phi_1 \, , \\
\end{cases}
\end{equation}
which allows us to freely specify only the initial density profile $\phi_0(x)$ and its time derivative $\phi_1(x)$.

In figure \ref{fig:soundone}, we compare the solution of \eqref{cauchysound} with those of \eqref{dispersive} obtained using different initial values for $\partial^2_t \phi$ and $\partial^3_t \phi$. As can be seen, any deviation from \eqref{cauchysound} leads to rapid oscillations with a characteristic frequency of order $\lambda^{-1}$. These oscillations are unphysical, since the effective field theory framework ceases to be valid at such high frequencies (with $\lambda$ representing the UV cutoff scale). Therefore, \eqref{cauchysound} provides the only physically meaningful initial data for the system.

It is worth noting that, in a non-dissipative system, an improper initialization of the additional degrees of freedom has far more severe consequences than in a dissipative one. In the latter case, dissipation tends to damp the fast modes, driving them rapidly to zero so that the system eventually relaxes toward a physical solution, albeit not the correct one (as illustrated in figure \ref{fig:heat}). In contrast, in a non-dissipative system, once these extra degrees of freedom are excited, they persist indefinitely (as illustrated in figure \ref{fig:soundone}).

\begin{figure}[b!]
    \centering
\includegraphics[width=0.43\linewidth]{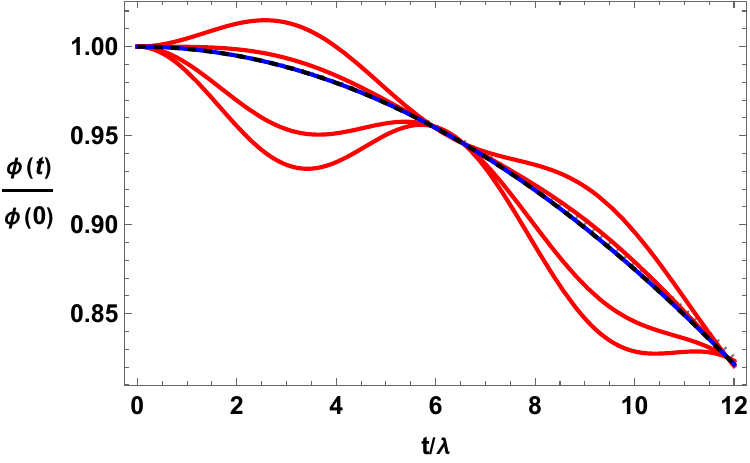}
\includegraphics[width=0.43\linewidth]{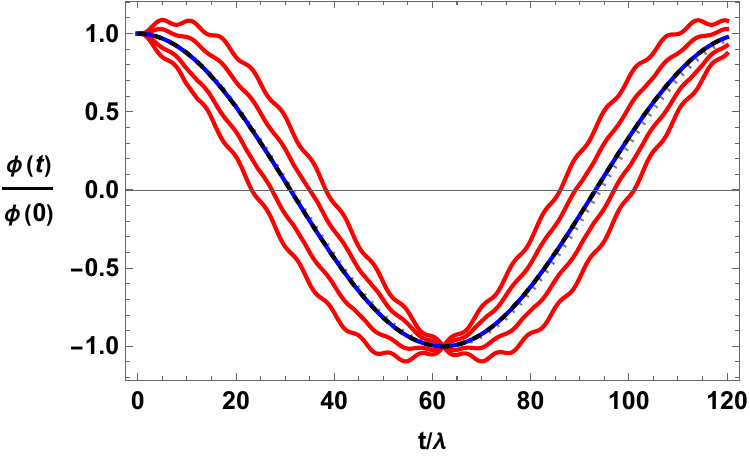}
\caption{Comparison between solutions of \eqref{dispersive} for various initial data (solid lines), under the ansatz $\phi(t,x)=\phi(t)e^{i k x}$, with $\lambda k=0.2$ and $c_s=1/2$. As initial conditions, we fix $\phi(0)$, and set $\partial_t \phi(0)=0$ in all cases. The blue curves show the solution of \eqref{cauchysound}. The red curves show solutions of \eqref{dispersive} obtained by varying the higher-order derivatives in the initial data. Specifically, in the left panel we vary $\partial_t^2\phi(0)\sim c_s^2 k^2$ while keeping $\partial_t^3 \phi(0)=0$, whereas in the right panel we fix $\partial_t^2\phi(0)$ according to \eqref{cauchysound} and vary $\partial_t^3 \phi(0)\sim c_s^2 k^2/\lambda$ (we just need $\lambda^3 \partial^3_t\ll1$ for validity of EFT, and $c_s^3k^3$ would be too small to be visible). All such solutions develop additional fast oscillations with frequency of order $\lambda^{-1}$, superimposed on the main infrared sound mode. In the right panel, changes of initial data also result in a global shift of phase of the underlying infrared mode. The only trajectory free from fast oscillations is the one corresponding to the Cauchy problem \eqref{cauchysound} (blue curve), which exactly overlaps with the solution of \eqref{dispersiveacausal} initialized with the same $\phi(0)$ and $\partial_t \phi(0)$ (black dashed). For reference, we also include the solution of $(\partial_t^2 - c_s^2\partial_x^2)\phi = 0$ with identical initial data (dotted gray line).}
\label{fig:soundone}
\end{figure}

\section{General considerations}

The general picture emerging from the above examples is as follows. When an effective field theory contains additional degrees of freedom, these correspond to fast modes (evolving on timescales comparable to the UV cutoff $\lambda$) and therefore lie outside the regime of validity of the theory. Consequently, the initial conditions must be chosen so that these fast modes are not excited. This can be achieved through the standard order-reduction procedure, which replaces the highest time derivatives with suitable combinations of spatial derivatives. The resulting equation retains the same level of accuracy as the original equation (the error scales with the same power of $\lambda$) but has fewer degrees of freedom. Thus, it can be used to fully constrain the extra initial data, automatically setting the fast modes to zero.

From the viewpoint of an effective field theorist, this is all rather obvious: in a low-energy effective theory, one should not construct initial states containing highly massive excitations. However, when a classical field theory is treated merely as a system of partial differential equations to be numerically solved with generic initial data, these consistency considerations are often overlooked (especially in the hydrodynamic context \cite{DisconziInitialData:2024dex}).

There are some general mathematical assumptions underlying the ``effective field theory reasoning'' above that warrant a more careful discussion, and shall be the topic of this section and the next one:
\begin{itemize}
\item Is it always true that \textit{all} the extra degrees of freedom are fast? This is discussed in Sec. \ref{sec:fast}.
\item Is it always possible to fix \textit{all} the extra degrees of freedom via order reduction? This is discussed in Sec. \ref{sec:Higher}.
\item Is the order-reduced equation always consistent with the original equation? This is discussed in Sec. \ref{sec:nosurv}.
\item Is the order-reduction procedure consistent with Lorentz invariance? This is discussed in Sec. \ref{sec:Lorentz}.
\end{itemize}
To keep the discussion fully rigorous, we shall specialize the first three questions to theories linearized about a constant background state. The non-linear case is discussed less formally in appendices (where feasible). 

\vspace{-0.2cm}
\subsection{The extra degrees of freedom are always fast}\label{sec:fast}
\vspace{-0.2cm}

The claim that ``all additional degrees of freedom evolve on a fast timescale'' can be formulated precisely and demonstrated for linear theories. Specifically, we have the following result.
\begin{theorem}\label{theo1}
Let $\phi:\textup{``Minkowski''}\rightarrow \mathbb{C}$ be a classical scalar field, governed by the linear partial differential equation
\begin{equation}\label{lineargeneral}
\mathcal{N}(\partial_t,\partial_j)\phi +\lambda \, \mathcal{M}(\partial_t,\partial_j)\phi=0\, ,
\end{equation}
where $\mathcal{N}(a,b_j)$ and $\mathcal{M}(a,b_j)$ are polynomials in $a$ and $b_j$ with constant coefficients. Call $\mathfrak{n}$ and $\mathfrak{m}$ the degrees of respectively $\mathcal{N}(a,b_j)$ and $\mathcal{M}(a,b_j)$ in the variable $a$ (and assume that it is independent of $b_j$). If $\mathfrak{m}\,{>}\,\mathfrak{n}$, then there are exactly $\mathfrak{m}{-}\mathfrak{n}$ dispersion relations (counting multiplicities) whose frequency tends to complex infinity as $\lambda\rightarrow 0$.
\end{theorem}
\begin{proof}
If we set $\phi\propto e^{ik_j x^j -i\omega t}$, equation \eqref{lineargeneral} reduces to
\begin{equation}\label{kippo}
\mathcal{N}(-i\omega,ik_j) +\lambda \, \mathcal{M}(-i\omega,ik_j)=0\, .
\end{equation}
Fixed the value of $k_j$, we can regard the left-hand side of \eqref{kippo} as a polynomial $\mathcal{P}_\lambda (\omega)$, which has degree $\mathfrak{m}$ when $\lambda\neq 0$, and has degree $\mathfrak{n}$ when $\lambda=0$. This can be written explicitly as follows (notice the factors $\lambda$):
\begin{equation}\label{ilpollo}
\lambda\, a_\mathfrak{m}(\lambda)\, \omega^\mathfrak{m}+\lambda\, a_{\mathfrak{m}-1}(\lambda)\, \omega^{\mathfrak{m}-1}+...+\lambda\, a_{\mathfrak{n}+1}(\lambda)\, \omega^{\mathfrak{n}+1}+ a_\mathfrak{n}(\lambda)\, \omega^\mathfrak{n}+ a_{\mathfrak{n}-1}(\lambda)\, \omega^{\mathfrak{n}-1}+...+ a_1(\lambda)\, \omega+ a_0(\lambda)=0 \, ,
\end{equation}
where $a_i(\lambda)$ are smooth functions, with $a_\mathfrak{n}(0)\neq 0$.
Dividing both sides by $\omega^\mathfrak{m}$, and setting $y=1/\omega$, we obtain
\begin{equation}\label{ilpollo2}
\lambda\, a_\mathfrak{m}(\lambda)+\lambda\, a_{\mathfrak{m}-1}(\lambda) \, y+...+\lambda\, a_{\mathfrak{n}+1}(\lambda)\, y^{\mathfrak{m}-\mathfrak{n}-1}+ a_\mathfrak{n}(\lambda)\, y^{\mathfrak{m}-\mathfrak{n}}+ a_{\mathfrak{n}-1}(\lambda)\, y^{\mathfrak{m}-\mathfrak{n}+1}+...+ a_1(\lambda)\, y^{\mathfrak{m}-1}+ a_0(\lambda) \, y^{\mathfrak{m}}=0 \, .
\end{equation}
Now, we observe that, when $\lambda=0$, the polynomial \eqref{ilpollo2} reduces to
\begin{equation}
y^{\mathfrak{m}-\mathfrak{n}} [a_\mathfrak{n}(0)+ a_{\mathfrak{n}-1}(0)\, y+...+ a_1(0)\, y^{\mathfrak{n}-1}+ a_0(0) \, y^{\mathfrak{n}}]=0 \, ,  
\end{equation}
and thus admits $y=0$ among its roots, with exact multiplicity $\mathfrak{m}-\mathfrak{n}$ (since $a_\mathfrak{n}(0)\neq 0$). By the standard theorem stating that polynomial roots depend continuously on the coefficients, it follows that for small $\lambda$, \eqref{ilpollo2} has precisely $\mathfrak{m}-\mathfrak{n}$ roots approaching $0$ as $\lambda\rightarrow 0$. Recalling that $\omega=1/y$, we then conclude that \eqref{ilpollo} has $\mathfrak{m}-\mathfrak{n}$ roots diverging to complex infinity, as claimed.
\end{proof}
The above theorem admits several easy generalizations. For example, one can take $[\mathcal{N}{+}\lambda \mathcal{M}_1{+}\lambda^2 \mathcal{M}_2{+}...]\phi=0$, and just define $\mathcal{M}(\lambda)=\mathcal{M}_1{+}\lambda\mathcal{M}_2{+}...$. Then, the proof proceeds more or less unchanged. The generalization to the case where $\phi$ is a list of tensors/spinors is also straightforward ($\mathcal{N}$ and $\mathcal{M}$ become matrices, and one needs to adjust the degree of \eqref{ilpollo} to account for the matrix size). In appendix \ref{appAAA}, we sketch an argument for the non-linear case.

\subsection{Reducing multiple orders}\label{sec:Higher}

As shown in the previous subsection, for a linear effective theory of the form  
\begin{equation}\label{fieldiamlo}
\mathcal{N}(\partial_t,\partial_j)\phi 
+ \lambda\, \mathcal{M}(\partial_t,\partial_j)\phi 
+ \mathcal{O}(\lambda^2) = 0\, ,
\end{equation}
among the $\mathfrak{m} \,{=}\, \text{deg}_a[\mathcal{M}(a,b_j)]$ total degrees of freedom, only  
$\mathfrak{n}\, {=}\, \text{deg}_a[\mathcal{N}(a,b_j)]$ correspond to slow (and thus physical) modes,  
whereas the remaining $\mathfrak{m}{-}\mathfrak{n}$ modes become infinitely fast in the limit of small $\lambda$. Therefore, for the initialization procedure to be well-defined, we must ensure the existence of an order-reduction scheme that can \textit{always} lower the order of \eqref{fieldiamlo} to $\mathfrak{n}$. Equivalently, we need to demonstrate that, by applying a suitable differential operator of the form  
$1+\lambda\,\mathcal{Q}(\partial_t,\partial_j)+\mathcal{O}(\lambda^2)$  
(with $\mathcal{Q}(a,b_j)$ a polynomial) to both sides of \eqref{fieldiamlo}, the resulting equation  
\begin{equation}\label{fieldiamlo33}
\mathcal{N}(\partial_t,\partial_j)\phi 
+ \lambda\, [\mathcal{M}(\partial_t,\partial_j)
+ \mathcal{Q}(\partial_t,\partial_j)\mathcal{N}(\partial_t,\partial_j)]\phi 
+ \mathcal{O}(\lambda^2) = 0
\end{equation}
can be made of order $\mathfrak{n}$ in time derivatives. As it turns out, the answer is affirmative, under very general assumptions.
\begin{theorem}\label{theo2}
Let $\mathcal{M}(a,b_j)$ and $\mathcal{N}(a,b_j)$ be polynomials in the variables $a$ and $b_j$. Assume that $\mathcal{M}$ and $\mathcal{N}$ have degrees $\mathfrak{m}$ and $\mathfrak{n}$ in $a$, respectively, with $\mathfrak{m}>\mathfrak{n}$. Suppose also that the coefficient of $a^{\mathfrak{n}}$ in $\mathcal{N}$ is equal to $1$. Then, there exists a polynomial $\mathcal{Q}(a,b_j)$ such that the combination $\mathcal{M}(a,b_j)+\mathcal{Q}(a,b_j)\mathcal{N}(a,b_j)$ has degree $\mathfrak{n}$ in $a$ or lower.
\end{theorem}
\begin{proof}
By isolating the terms of highest degree in $a$ from the two polynomials, we can write  
\begin{equation}
\begin{split}
\mathcal{N}(a,b_j) &= a^{\mathfrak{n}} + \mathcal{R}(a,b_j)\,, \\
\mathcal{M}(a,b_j) &= \mathcal{P}_0(b_j)a^{\mathfrak{m}} + \mathcal{R}_0(a,b_j)\, ,
\end{split}
\end{equation}
where the remainder terms $\mathcal{R}$ and $\mathcal{R}_0$ are polynomials whose degrees in $a$ do not exceed $\mathfrak{n}{-}1$ and $\mathfrak{m}{-}1$, respectively.  
We now define
\begin{equation}
\mathcal{M}_1(a,b_j) = \mathcal{M}(a,b_j) - \mathcal{P}_0(b_j)a^{\mathfrak{m}-\mathfrak{n}}\mathcal{N}(a,b_j)
= \mathcal{R}_0(a,b_j) - \mathcal{P}_0(b_j)a^{\mathfrak{m}-\mathfrak{n}}\mathcal{R}(a,b_j)\, .
\end{equation}
This new polynomial $\mathcal{M}_1$ has degree $\mathfrak{m}_1 \leq \mathfrak{m}{-}1$.  
If $\mathfrak{m}_1 \leq \mathfrak{n}$, the reduction is complete, and $\mathcal{Q}(a,b_j) = - \mathcal{P}_0(b_j)a^{\mathfrak{m}-\mathfrak{n}}$. Otherwise, we again isolate the highest-order term:
\begin{equation}
\mathcal{M}_1(a,b_j) = \mathcal{P}_1(b_j)a^{\mathfrak{m}_1} + \mathcal{R}_1(a,b_j)\, ,
\end{equation}
where $\mathcal{R}_1$ has $a-$degree at most $\mathfrak{m}_1{-}1$.  
We then define
\begin{equation}
\mathcal{M}_2(a,b_j) = \mathcal{M}_1(a,b_j) - \mathcal{P}_1(b_j)a^{\mathfrak{m}_1-\mathfrak{n}}\mathcal{N}(a,b_j)
= \mathcal{R}_1(a,b_j) - \mathcal{P}_1(b_j)a^{\mathfrak{m}_1-\mathfrak{n}}\mathcal{R}(a,b_j)\, .
\end{equation}
This iterative process is repeated until the $a-$degree of the resulting polynomial $\mathcal{M}_i$ reaches $\mathfrak{n}$ or falls below it.  
The final expression for the desired polynomial is therefore
$\mathcal{Q}(a,b_j) = -\sum_i \mathcal{P}_i(b_j)a^{\mathfrak{m}_i-\mathfrak{n}}$.
\end{proof}
Also this theorem admits several generalizations, e.g. to the case where $\mathcal{N}$ and $\mathcal{M}$ are matrices. 

In appendix \ref{AppBBB}, we briefly discuss the non-linear case.

A natural follow-up question concerns the uniqueness of the order-reduced equation. 
In general, it is not unique. 
Indeed, if $\mathcal{Q}(a,b_j)$ satisfies the assumptions of Theorem~\ref{theo2}, 
then so does $\tilde{\mathcal{Q}}(a,b_j)=\mathcal{Q}(a,b_j)+\mathcal{F}(b_j)$. 
This ambiguity is physically immaterial, as any two order-reduced equations derived from the same effective field theory 
agree with the underlying EFT up to the truncation error, and are therefore 
equivalent to each other. 
In particular, adding a polynomial $\mathcal{F}(b_j)$ to $\mathcal{Q}$ corresponds, within the EFT 
accuracy, to multiplying both sides of \eqref{fieldiamlo33} by 
$1+\lambda \mathcal{F}(b_j)+\mathcal{O}(\lambda^2)$. In practice, however, it is convenient to fix $\mathcal{F}$ so that $\mathcal{M}{+}\Tilde{\mathcal{Q}}\mathcal{N}$ has degree \textit{strictly less} than $\mathfrak{n}$ (by taking $-\mathcal{F}$ to be the coefficient of $a^\mathfrak{n}$ in $\mathcal{M}{+}\mathcal{Q}\mathcal{N}$), thereby making the initialization prescriptions for $(\partial_t^{\mathfrak{n}}\phi)_{t=0},...,(\partial_t^{\mathfrak{m}-1}\phi)_{t=0}$ algebraic.

\vspace{-0.3cm}
\subsection{Consistency of the order reduction}\label{sec:nosurv}
\vspace{-0.2cm}

The claim that the original and the order-reduced equations are ``equivalent at the given order in $\lambda$'' is purely formal, as it implicitly assumes that all physical observables depend smoothly on $\lambda$, and can be Taylor-expanded. Let us see how such an equivalence statement can be made rigorous in the case of linear equations in a constant background.
\begin{theorem}
Let $\phi:\textup{``Minkowski''}\rightarrow \mathbb{C}$ be a classical scalar field, and consider two alternative differential equations:
\begin{equation}\label{lineargeneralIII}
\begin{split}
&\mathcal{N}(\partial_t,\partial_j)\phi +\lambda \, \mathcal{M}(\partial_t,\partial_j)\phi=0\, ,\\
&\mathcal{N}(\partial_t,\partial_j)\phi 
+ \lambda\, [\mathcal{M}(\partial_t,\partial_j)
{+} \mathcal{Q}(\partial_t,\partial_j)\mathcal{N}(\partial_t,\partial_j)]\phi =0\, ,\\
\end{split}
\end{equation}
where $\mathcal{N}(a,b_j)$, $\mathcal{M}(a,b_j)$, and $\mathcal{Q}(a,b_j)$ are polynomials in $a$ and $b_j$ with constant coefficients. Suppose that all the roots $a$ of $\mathcal{N}(a,b_j)$ have multiplicity 1 for almost all $b_j$. Then, the slow dispersion relations of the two theories in \eqref{lineargeneralIII} agree to first order in $\lambda$.
\end{theorem}
\begin{proof}
If we set $\phi\propto e^{ik_j x^j-i\omega t}$, the first line of \eqref{lineargeneralIII} reduces to 
\begin{equation}\label{dispuzio}
\mathcal{N}(-i\omega) +\lambda \, \mathcal{M}(-i\omega)=0\, .
\end{equation}
Here, we have suppressed the dependence on $k_j$ by fixing its value. Such value has been chosen so that $\mathcal{N}(-i\omega)$ possesses only simple roots (the exceptional points where this condition fails are handled automatically by continuity). Now, we can view equation \eqref{dispuzio} as an implicit function $F(\omega,\lambda)=0$. Hence, let $\omega_0$ be a root of $\mathcal{N}(-i\omega)$, and observe the following facts: (a) $F(\omega_0,0)=0$, (b) $\partial_\omega F(\omega_0,0)=-i\mathcal{N}'(-i\omega_0)\neq 0$ (since $\omega_0$ is a simple root), and (c) $F(\omega,\lambda)$ is an entire function. Then, the holomorphic implicit function theorem applies, and there exists a unique local solution $\omega(\lambda)$ of \eqref{dispuzio} such that $\omega(0)=\omega_0$. Such a solution is analytic in a neighborhood of $\lambda=0$, and thus admits an expansion of the form $\omega(\lambda)=\omega_0+\lambda\omega_1+\mathcal{O}(\lambda^2)$, where $\omega_1$ is determined by the implicit function theorem to be
\begin{equation}\label{sguiz}
\omega_1 = -\dfrac{\partial_\lambda F(\omega_0,0)}{\partial_\omega F(\omega_0,0)}=\dfrac{\mathcal{M}(-i\omega_0)}{i\mathcal{N}'(-i\omega_0)}\, .
\end{equation}
If we carry out the same calculation starting from the second line of \eqref{lineargeneralIII}, we arrive at an identical result, except that in \eqref{sguiz} one would replace $\mathcal{M}(-i\omega_0)$ with $\mathcal{M}(-i\omega_0)+\mathcal{Q}(-i\omega_0)\mathcal{N}(-i\omega_0)$. However, since $-i\omega_0$ is a root of $\mathcal{N}$, the additional term $\mathcal{Q}(-i\omega_0)\mathcal{N}(-i\omega_0)$ vanishes, so the resulting values of $\omega_1$ are identical. 

In summary, we have shown that, for every root of $\mathcal{N}$, there exists a corresponding slow mode whose first-order expansion in $\lambda$ yields the same result for both equations in \eqref{lineargeneralIII}. But since $\mathcal{N}$ has only simple roots, the number of such slow modes equals the degree of $\mathcal{N}$, and Theorem \ref{theo1} reassures us that all the remaining modes are fast.
\end{proof}
This theorem allows for generalizations analogous to those discussed previously. The only substantive restriction here is the assumption that $\mathcal{N}$ possesses only simple roots. When this condition fails, the Taylor expansion must be replaced by a Puiseux expansion. In that case, one finds that the leading order $\lambda^{1/r}$  of that expansion remains identical for the two theories, in all branches. Also, note that a similar reasoning applies at higher orders in $\lambda$, since the holomorphic implicit function theorem can determine all derivatives with respect to $\lambda$ at $0$.

\vspace{-0.3cm}
\section{Lorentz and general covariance}\label{sec:Lorentz}
\vspace{-0.2cm}

Our proposal for how to eliminate unphysical degrees of freedom makes use of the reduction of order procedure. This procedure is not used to modify equations of motion, but only to set physically meaningful initial data. However, this still involves breaking formal
Lorentz invariance (or general covariance). Nevertheless, we will argue in this section that, if one restricts to Lorentz frames (or coordinate systems) in which EFT is manifestly valid, then the effects of this breaking are of the same size as the effects of EFT terms which are neglected in the original equation of motion. Hence, the proposal does not introduce any \textit{unphysical} breaking of Lorentz invariance.

\vspace{-0.4cm}
\subsection{Preservation of Lorentz covariance}
\vspace{-0.3cm}

\label{subsec:Lorentz}

Let's start by considering a Lorentz covariant theory. The key assumption underlying any EFT is that the characteristic gradient length scale \( L \) is much larger than the microscopic UV cutoff length scale \( \lambda \). Assume that we have a well-posed and Lorentz invariant system of equations describing an EFT truncated at some order \( \mathfrak{m} \), i.e., the neglected higher order terms in the equation of motion scale as \( (\lambda / L)^{\mathfrak{m}+1} \). Consider an inertial observer Alice for whom the validity of EFT assumption holds, at least initially. She can perform the reduction of order procedure to determine initial data for higher derivatives of fields at $t=0$, modulo EFT error terms of order \( (\lambda / L)^{\mathfrak{m}+1} \). She can then neglect these error terms to fix the initial data for higher derivatives exactly. Well-posedness then guarantees the existence of a unique solution of the equations of motion. 

Now consider this from the perspective of another inertial observer, Bob. If Bob moves at a moderate speed relative to Alice (say, \( v \sim 0.5 \)), the typical gradient length scale remains approximately \( L \), so the error in his frame also scales as \( (\lambda / L)^{\mathfrak{m}+1} \). Hence, the Lorentz transformation of Alice's solution will satisfy the validity condition of the EFT also in Bob's frame. Bob can then apply reduction of order in his frame to specify initial data for higher derivatives of fields at time $t'=0$. This procedure uses only the EFT equation of motion, which (the Lorentz transformation of) Alice's solution satisfies in Bob's frame. Therefore the initial data for higher derivatives imposed by order reduction in Bob's frame may differ from the exact initial data arising from (the Lorentz transformation of) Alice's solution only by an EFT error term of order \( (\lambda / L)^{\mathfrak{m}+1} \). Finally we may invoke the ``continuous dependence'' (Cauchy stability) aspect of well-posedness to deduce that, at least in a compact region of spacetime, the difference between the solution arising from Bob's initial data and (the Lorentz transformation of) Alice's solution is of order \( (\lambda / L)^{\mathfrak{m}+1} \). This is the same as the expected size of the deviation arising from the higher order EFT terms that were neglected in the original equation of motion. Notice how this argument makes use of both the Lorentz covariance and the well-posedness of the equation of motion. 

Consider now the case where Bob moves with ultrarelativistic speed relative to Alice ($v \to 1$). 
In this regime, Lorentz contraction becomes significant, and the characteristic gradient length scale in Bob's frame is
\(L' = L/\gamma\).
If $\gamma$ is sufficiently large, one may have $L' \lesssim \lambda$, so that Bob can no longer justify the applicability of the effective field theory on the basis of a local gradient expansion.
Nevertheless, suppose that Bob is informed by Alice that, in her rest frame, the EFT assumptions are satisfied. 
In this case, Bob may still consistently rely on the Lorentz-invariant effective theory, since its solutions coincide with the Lorentz transformation of Alice’s solutions, which are legitimately obtained within the EFT regime\footnote{An analogous situation arises in particle physics. Any given particle can have arbitrarily large energy in the rest frame of a suitably boosted observer. Nevertheless, if all particles involved in an interaction have energies well below the UV cutoff in the center-of-momentum frame, the process is consistently described by a low-energy effective theory formulated in that frame. An observer moving ultrarelativistically with respect to the center-of-momentum frame cannot justify the use of the EFT via a local energy expansion in their own frame. However, if the EFT is Lorentz invariant, they may still use it operationally, since this is equivalent to performing the computation in the center-of-momentum frame (where the EFT holds) and subsequently boosting the result.}.
Crucially, however, Bob is not allowed to perform Lorentz-violating operations such as order reductions, because his use of the effective theory is not grounded in a controlled expansion in gradients (which are large in his frame), but rather in the existence of another observer (Alice) for whom the EFT is manifestly valid. 
The only consistent procedure available to Bob for specifying initial data for the Lorentz-invariant EFT is to take Alice’s solution and extract from it the values of the fields and their derivatives on Bob’s simultaneity surface.

The criterion that separates observers who are allowed to perform order reduction from those who are not is
$\gamma \lambda/L\ll 1$. 
When this condition is violated, Bob has no justification for using order reduction to set up initial data.
As a quantitative illustration, suppose that $\lambda/L \sim 10^{-3}$ in Alice’s frame, and that the gradient expansion ceases to be reliable when $\gamma \lambda/L \sim 10^{-1}$. 
It then follows that Bob may perform order reduction only if $|v| \lesssim 0.99995$.

\vspace{-0.4cm}
\subsection{Preservation of general covariance}\label{genercovuz} 
\vspace{-0.3cm}

We will now sketch the generalization of this argument to gravitational theories. For simplicity we restrict to vacuum gravity, i.e., no matter fields (the problem of BDNK in curved spacetime is outlined at the end of Sec. \ref{bdnksection}). We will argue that the order reduction procedure preserves general covariance modulo the truncation error of the EFTs (see section \ref{sec:EGB} for more details on the order reduction in gravitational EFTs). In a gravitational EFT with a UV length scale $\lambda$, the condition for validity of EFT is that the metric is {\it slowly varying}, meaning that there exists a coordinate system in which $j$th derivatives of metric components are of order $L^{-j}$ with $L \gg \lambda$. 

The initial data for an EFT of vacuum gravity in $\mathcal{D}$ spacetime dimensions, truncated at some order $\mathfrak{m}+2$ in derivatives, is of the form $(\Sigma,\gamma,K, \pounds_n K, \ldots, \pounds_n^\mathfrak{m} K)$ where $\Sigma$ is a manifold of dimension $d=\mathcal{D}-1$, $\gamma$ a Riemannian metric on $\Sigma$, and $\pounds_n^j K$ ($0 \le j \le \mathfrak{m})$ are symmetric tensors on $\Sigma$. This data should satisfy the constraint equations of the EFT.
We assume existence of a well-posed, and generally covariant, formulation of the equations of this EFT, for example given by the approach of FHK \cite{Figueras:2024bba}. This means the following: (a) there exists a spacetime $(M,g)$ satisfying the EFT equation of motion and for which $\Sigma$ can be identified with a hypersurface in $M$ with induced metric $\gamma$, extrinsic curvature $K$, and the $j$th Lie derivative of the extrinsic curvature along the unit normal $n^a$ of $\Sigma$ are $\pounds_n^j K$, $1 \le j \le \mathfrak{m}$; (b) the solution $(M,g)$ is determined uniquely up to diffeomorphisms\footnote{This should be understood modulo the freedom to change $M$. It might be possible to fix $M$ uniquely by introducing the notion of a maximal globally hyperbolic development but existence of such a development has not yet been studied in the FHK formalism.} that reduce to the identity on $\Sigma$; (c) the solution depends continuously on the initial data in the sense that gauge equivalence classes of solutions depend continuously on gauge equivalence classes of initial data \cite{Hawking1973}.

We assume that the error introduced in truncating the EFT at order $\mathfrak{m}+2$ is of order\footnote{
More precisely, it is of order $\lambda^{\mathfrak{m}+1}/L^{\mathfrak{m}+3}$ since Einstein's equation has dimensions $L^{-2}$, but we won't include powers of $L$ now.} $\lambda^{\mathfrak{m}+1}$. Now assume we are given a triple $(\Sigma_A,\gamma_A,K_A)$ with $(\gamma_A,K_A)$ slowly varying in some coordinate chart on $\Sigma_A$. We can apply the order reduction procedure to determine $\pounds_n^j K$, $1 \le j \le k$ in terms of $(\gamma_A,K_A)$ modulo error terms of order $\lambda^{\mathfrak{m}+1}$ which we neglect. The EFT constraint equations then reduce to constraint equations for $(\gamma_A,K_A)$ (involving higher spatial derivatives) which we assume can be solved (see Sec. \ref{sec:EGB} for further discussion). We then appeal to (a) to obtain a solution $(M_A,g_A)$ of our EFT equation of motion. We assume this is slowly varying, if it is not then we shrink $M_A$ until it is. Now apply a diffeomorphism $\varphi$ which maps this to $(\varphi(M_A),\varphi_\star g_A)$. This is also a solution of our truncated EFT because the equations of motion are generally covariant. We assume that the diffeomorphism is slowly varying, i.e., there exist coordinates such that $(\varphi(M_A),\varphi_\star g_A)$ satisfies the EFT assumptions. Pick a spacelike hypersurface $\Sigma_B$ in the latter spacetime. Let's denote its initial data as $(\Sigma_B,\gamma_B,K_B, \pounds_n^j K_B)$. (This is analogous to the data of the Lorentz transformation of Alice's solution in section \ref{subsec:Lorentz}.) We now apply the reduction of order procedure on this surface. 
Starting from $(\Sigma_B,\gamma_B,K_B)$ this determines data for the higher order derivatives $\overline{ \pounds_n^j K_B}$ ($1 \le j \le \mathfrak{m}$), where we use an overbar to indicate that these may differ from $\pounds_n^j K_B$ by error terms of order $\lambda^{\mathfrak{m}+1}$. Since $(\varphi(M_A),\varphi_\star g_A)$ is a solution, it satisfies the constraint equations, and hence $(\Sigma_B,\gamma_B,K_B,\overline{ \pounds_n^j K_B})$ satisfies the constraints modulo an error of order $\lambda^{\mathfrak{m}+1}$. We assume that this can be improved to an exact solution of the constraints by an adjustment to $(\gamma_B,K_B)$ of order $\lambda^{\mathfrak{m}+1}$, retaining the same expressions for $\overline{ \pounds_n^j K_B}$ in terms of $(\gamma_B,K_B)$. From (a) the data $(\Sigma_B,\gamma_B,K_B,\overline{ \pounds_n^j K_B})$ determines a solution $(M_B,g_B)$ with $\Sigma_B \subset M_B$. This initial data differs from $(\Sigma_B,\gamma_B,K_B,\pounds_n^j K_B)$ only by terms of order $\lambda^{\mathfrak{m}+1}$ hence by (c) we deduce that $g_B$ is gauge equivalent to $\varphi_\star g_A$ modulo terms of order $\lambda^{\mathfrak{m}+1}$.

\section{Applications}

We conclude this article by discussing some concrete applications. Specifically, we will examine a few non-linear causal, stable, and well-posed effective field theories of interest that possess extra degrees of freedom, and we will outline the appropriate initialization prescriptions for each.

\subsection{Causal heat conduction in rotating stars}

The standard theory for heat transport in rigid bodies reads \cite{Schaab:1998ed,Negreiros2012,Potekhin:2015qsa,Beznogov:2022wae}
\begin{equation}\label{CuT}
C u^\mu \nabla_\mu T_R=\lambda \nabla_\mu [\kappa (g^{\mu \nu}{+}u^\mu u^\nu)\nabla_\nu T_R]+\mathcal{O}(\lambda^2) \, ,   
\end{equation}
where $C$ is the specific heat, $u^\mu$ is the flow velocity, $T_R$ is the redshifted temperature, and $\kappa$ is the heat conductivity. The coefficient $\lambda$ serves as a bookkeeping parameter, used to track the order of each term in the mean free path expansion. Unfortunately, equation \eqref{CuT} is acausal and becomes unstable and ill-posed in rotating systems \cite{Hiscock_Insatibility_first_order,Kost2000,GavassinoSuperlum2021}. Therefore, it was proposed in \cite{GavassinoAntonelli:2025umq} to replace \eqref{CuT} with a BDNK-type alternative:
\begin{equation}\label{CuTcausal}
C u^\mu \nabla_\mu T_R=\lambda \nabla_\mu (\kappa \nabla^\mu T_R)+\mathcal{O}(\lambda^2)  \, .  
\end{equation}
Let us see how this equation should be initialised.

Having neutron-star applications in mind, we consider a stationary and axially symmetric background spacetime, and use quasi-isotropic spherical coordinates
$\{t,\phi,r,\theta\}$ \cite{Beznogov:2022wae}. In these coordinates, the components $g_{\mu\nu}(r,\theta)$ of the metric and the components $g^{\mu \nu}(r,\theta)$ of the inverse metric read
\begin{equation}
g_{\mu \nu}=
\begin{bmatrix}
\rho^2\omega^2{-}e^{2\Phi} & -\omega \rho^2 & 0 & 0 \\
-\omega \rho^2 & \rho^2 & 0 & 0 \\
0 & 0 & e^{2\Psi} & 0 \\
0 & 0 & 0 & e^{2\Psi} r^2 \\
\end{bmatrix} \, , \spc
g^{\mu \nu}=
\begin{bmatrix}
-e^{-2\Phi} & -\omega e^{-2\Phi} & 0 & 0 \\
-\omega e^{-2\Phi} & \rho^{-2}{-}\omega^2 e^{-2\Phi} & 0 & 0 \\
0 & 0 & e^{-2\Psi} & 0 \\
0 & 0 & 0 & e^{-2\Psi} r^{-2} \\
\end{bmatrix}\, ,
\end{equation}
where
$\{\Phi,\Psi,\rho,\omega\}$ are some metric functions. In these coordinates, a rigidly rotating star has collective flow velocity $u^\mu \partial_\mu =e^{-\Phi}\gamma(\partial_t \,{+}\,\Omega\partial_\phi)$, with $\Omega=\text{const}$ the rotation frequency and
\begin{equation}
\gamma=\dfrac{1}{\sqrt{1-e^{-2\Phi}\rho^2 (\Omega{-}\omega)^2 }}
\end{equation}
the Lorentz factor of the medium relative to a Zero Angular Momentum Observer (ZAMO). Then, the redshifted temperature is just
$T_R=e^\Phi T/\gamma$, and equation \eqref{CuTcausal} explicitly reads
\begin{equation}
C e^{-\Phi}\gamma(\partial_t \,{+}\,\Omega\partial_\phi) T_R=\dfrac{\lambda}{\sqrt{-g}} \partial_\mu (\kappa \sqrt{-g}g^{\mu \nu}\partial_\nu T_R)+\mathcal{O}(\lambda^2) \, ,   
\end{equation}
with
$\sqrt{-g}=e^{\Phi+2\Psi}\rho r$. This equation is of second order in time, so we need to assign the profiles of both $T$ and $\partial_t T$ at $t=0$. To constrain the latter in terms of the former, we follow the same procedure as in section \ref{sec:rod}, with the only caveat that $\kappa=\kappa(r,\theta,T_R)$ depends on $T_R$ itself, and thus the replacement of the external derivative $\partial_t$ with $-\Omega\partial_\phi$ must be handled with extra care. Below, we report the final prescription for the full Cauchy problem truncated to order $\lambda$:
\begin{equation}\label{neutroncauchy}
\begin{cases}
(\partial_t \,{+}\,\Omega\partial_\phi)T_R= \dfrac{\lambda e^\Phi}{C\gamma \sqrt{-g}} \partial_\mu (\kappa\sqrt{-g} \,g^{\mu\nu}\partial_\nu T_R) \, , \\
T_R|_{t=0}=T_R^o\, ,\\
\partial_t T_R|_{t=0}=-\Omega\partial_\phi T_R^o+ \dfrac{\lambda e^\Phi}{C \gamma} \left[ \dfrac{1}{\sqrt{-g}} \partial_r \left(\kappa\sqrt{-g} \,g^{rr}\partial_r T_R^o\right)+\dfrac{1}{\sqrt{-g}} \partial_\theta \left(\kappa\sqrt{-g} \,g^{\theta\theta}\partial_\theta T_R^o\right)+\dfrac{1}{\rho^2 \gamma^2}\partial_\phi \left(\kappa \partial_\phi T_R^o\right)\right]\, .\\
\end{cases}
\end{equation}
Once again, we note that, if preserving causality and covariance in an exact fashion is not a concern, one may just elevate the third line of \eqref{neutroncauchy} to a full equation of motion (which happens to be itself well-posed):
\begin{equation}
(\partial_t+\Omega\partial_\phi) T_R= \dfrac{\lambda e^\Phi}{C \gamma} \left[ \dfrac{1}{\sqrt{-g}} \partial_r \left(\kappa\sqrt{-g} \,g^{rr}\partial_r T_R\right)+\dfrac{1}{\sqrt{-g}} \partial_\theta \left(\kappa\sqrt{-g} \,g^{\theta\theta}\partial_\theta T_R\right)+\dfrac{1}{\rho^2 \gamma^2}\partial_\phi \left(\kappa \partial_\phi T_R\right)\right]+\mathcal{O}(\lambda^2)\, .  
\end{equation}
This would correspond to a ``density-frame'' formulation of heat conduction in rotating stars \cite{Armas:2020mpr,Basar:2024qxd,GavassinoParabolic2025hwz}.

\subsection{BDNK}\label{bdnksection}
\vspace{-0.1cm}

It is finally time to discuss the proper initialization of BDNK (in Minkowski spacetime, for simplicity). Consider a relativistic viscous fluid at zero chemical potential, and introduce, as usual, a bookkeeping parameter $\lambda$ that keeps track of the order in the mean-free-path expansion. Then, the BDNK equations of motion for the inverse-temperature four-vector $\beta_\mu$ ($=$``flow velocity''/``temperature'') can be written in the form \cite{Clarisse:2025lli}
\begin{equation}\label{BDNKEOM}
T^{\mu \nu}=T^{\mu \nu}_{\text{Ideal}}+\lambda H^{\mu \nu \rho\sigma}\partial_{\rho}\beta_{\sigma}+\mathcal{O}(\lambda^2) \, ,
\end{equation}
where $T^{\mu \nu}_{\text{Ideal}}(\beta_\rho)$ is the stress-energy tensor of the ideal fluid, and $H^{\mu \nu \rho\sigma}=H^{(\mu \nu) (\rho\sigma)}(\beta_\theta)$ is the susceptibility of the energy-momentum tensor to spacetime gradient. Then, the conservation law $\partial_\nu T^{\mu\nu}=0$ explicitly reads
\begin{equation}
M^{\mu \nu \rho} \partial_\nu \beta_\rho +\lambda H^{\mu \nu \rho\sigma} \partial_\nu \partial_{\rho}\beta_{\sigma}+\lambda K^{\mu \nu\rho\sigma\theta} \partial_\nu \beta_\theta\, \partial_\rho \beta_\sigma  +\mathcal{O}(\lambda^2)=0\, ,
\end{equation}
with
\begin{equation}
M^{\mu \nu \rho}= \dfrac{\partial T_{\text{Ideal}}^{\mu \nu}}{\partial \beta_\rho} \, , \spc
K^{\mu \nu\rho\sigma\theta}= \dfrac{\partial H^{\mu \nu\rho\sigma}}{\partial \beta_\theta} \, . 
\end{equation}
Let us separate the time derivatives from the space derivatives:
\begin{equation}\label{thelong}
\begin{split}
& M^{\mu t \rho} \partial_t \beta_\rho+M^{\mu j \rho} \partial_j \beta_\rho \\
& +\lambda H^{\mu tt\sigma} \partial_t \partial_{t}\beta_{\sigma}+2\lambda H^{\mu (jt)\sigma} \partial_j \partial_{t}\beta_{\sigma}+\lambda H^{\mu jk\sigma} \partial_j \partial_{k}\beta_{\sigma}\\
&+\lambda K^{\mu tt\sigma\theta} \partial_t \beta_\theta\, \partial_t \beta_\sigma +\lambda K^{\mu tk\sigma\theta} \partial_t \beta_\theta\, \partial_k \beta_\sigma +\lambda K^{\mu jt\sigma\theta} \partial_j \beta_\theta\, \partial_t \beta_\sigma +\lambda K^{\mu jk\sigma\theta} \partial_j \beta_\theta\, \partial_k \beta_\sigma  +\mathcal{O}(\lambda^2)=0\, .\\
\end{split}
\end{equation}
To zeroth order in $\lambda$, we have the vectorial equation $M^{t } \partial_t \beta+M^{ j } \partial_j \beta=\mathcal{O}(\lambda)$, where we have defined the $4\times 4$ matrices $M^\nu=[M^{\mu \nu \rho}]$. Multiplying both sides by $(M^t)^{-1}$, we obtain the equation
\begin{equation}\label{dtbetone}
\partial_t \beta_\rho=L_\rho^{j \xi} \partial_j \beta_\xi +\mathcal{O}(\lambda) \, ,
\end{equation}
with $L^j=-(M^t)^{-1} M^j$. Our task, then, is to use the above equation to replace time derivatives with space derivatives in all terms with a prefactor $\lambda$, since the error that we make can be reabsorbed into $\mathcal{O}(\lambda^2)$. Thus, equation \eqref{thelong} becomes
\begin{equation}\label{theuglymonster}
\begin{split}
& M^{\mu t \rho}( \partial_t \beta_\rho-L_\rho^{j \xi} \partial_j \beta_\xi) \\
& +\lambda H^{\mu tt\sigma} \dfrac{\partial L_\sigma^{j \xi}}{\partial\beta_\alpha} L_\alpha^{k \zeta} \partial_k \beta_\zeta \partial_j \beta_\xi+\lambda H^{\mu tt\sigma} L_\sigma^{j \xi} \partial_j (L_\xi^{k \zeta} \partial_k \beta_\zeta)+2\lambda H^{\mu (jt)\sigma} \partial_j (L_\sigma^{k \xi} \partial_k \beta_\xi)+\lambda H^{\mu jk\sigma} \partial_j \partial_{k}\beta_{\sigma}\\
&+\lambda K^{\mu tt\sigma\theta} L_\theta^{j \xi} \partial_j \beta_\xi\, L_\sigma^{k \zeta} \partial_k \beta_\zeta +\lambda K^{\mu tk\sigma\theta} L_\theta^{j \xi} \partial_j \beta_\xi\, \partial_k \beta_\sigma +\lambda K^{\mu jt\sigma\theta} \partial_j \beta_\theta\, L_\sigma^{j \xi} \partial_j \beta_\xi +\lambda K^{\mu jk\sigma\theta} \partial_j \beta_\theta\, \partial_k \beta_\sigma  =\mathcal{O}(\lambda^2)\, .\\
\end{split}
\end{equation}
This is the equation that should be used to correctly initialize BDNK. Given its complexity, it is tempting to opt for the easier prescription \eqref{dtbetone} (as is done in \cite{Bea:2023rru}), which may be accurate enough for applications where the initial data is anyway not precisely known. However, to attain order-$\lambda$ precision (which is the precision goal of BDNK theory) for known physical data, it is necessary to employ \eqref{theuglymonster}.

The above initialization procedure extends straightforwardly to curved spacetimes. In this case, Christoffel-symbol corrections enter \eqref{thelong}, generating additional contributions in \eqref{theuglymonster}. A subtlety arises if the spacetime is dynamical, since the initial data must also satisfy the constraint equations $G^{\mu 0}=8\pi T^{\mu 0}$. The energy-momentum tensor \eqref{BDNKEOM} contains time derivatives $\partial_t \beta_\rho$, but these appear multiplied by $\lambda$. Therefore, consistently within the EFT accuracy, we can replace these time derivatives with either \eqref{theuglymonster} or just with \eqref{dtbetone} (supplemented with appropriate connection terms), since in both cases the error to $T^{\mu 0}$ is of second order.

The preservation of general covariance at first order in $\lambda$ follows the same logic as in Sec.~\ref{genercovuz}. In particular, adopting the same notation, the initial data for the Einstein-BDNK system is $(\Sigma,\gamma,K,T,u,\pounds_n T,\pounds_n u)$. The order-reduction procedure uniquely determines $\pounds_n T_A$ and $\pounds_n u_A$ in terms of a given configuration $(\Sigma_A,\gamma_A,K_A,T_A,u_A)$. If we now apply a slowly varying diffeomorphism to the corresponding solution, we obtain a new spacetime from which we can extract data $(\Sigma_B,\gamma_B,K_B,T_B,u_B,\pounds_n T_B,\pounds_n u_B)$. The values of $\pounds_n T_B$ and $\pounds_n u_B$ coincide with those prescribed by order reduction up to errors of order $\lambda^2$. By well-posedness, this implies that the associated solution is equivalent to the original one modulo $\mathcal{O}(\lambda^2)$ corrections.

\vspace{-0.2cm}
\subsection{Regularised Einstein-Gauss-Bonnet gravity}\label{sec:EGB}
\vspace{-0.1cm}

The order reduction scheme can also be used to fix initial data for the fast degrees of freedom in gravitational EFTs. There is, however, an added complication associated with the diffeomorphism-invariance of the action: some components of the equations of motion are initial value constraints (similar issues arise in gauge theories too, such as in electromagnetism). In this section, we explain how the order reduction procedure applies to such theories.

More specifically, consider a gravitational theory with a diffeomorphism-invariant Lagrangian that is locally constructed out of the metric tensor, the Riemman tensor and its covariant derivatives. Assume that the equation of motion of this theory,
\begin{equation}
    E^{\mu \nu}=0,
\end{equation}
is of order $\mathfrak{m}+2$ in derivatives of the metric. Cauchy data for such a theory consist of a Cauchy slice $\Sigma$ with unit normal $n$, the induced metric {$\gamma_{\mu \nu}=g_{\mu \nu}+n_\mu n_\nu$} on $\Sigma$, the extrinsic curvature $K_{\mu \nu}$ and its normal derivatives $(\pounds_n)^k K_{\mu \nu}$ up to order $k\leq \mathfrak{m}$. It follows from the diffeomorphism symmetry that the equations $E^{\mu \nu} n_\nu=0$ are initial value constraints that the Cauchy data must satisfy.\footnote{A simple way to see that $E^{\mu \nu} n_\nu=0$ are indeed constraints goes as follows. Diffeomorphism-invariance of the action implies the \textit{off-shell} Bianchi identity $\nabla_\nu E^{\mu\nu}=0$. Decomposing these equations to components parallel and orthogonal to $n$ shows that normal (time) derivatives of $E^{\mu \nu} n_\nu$ are related to spatial derivatives of the evolution equations (i.e. the orthogonal projections of $E_{\mu\nu}$). This means that the $E^{\mu \nu} n_\nu$ components contain $1$ fewer time derivatives of the metric than $E^{\mu \nu} \gamma^\alpha_\mu \gamma^\beta_\nu$ and hence are initial data constraints. An alternative argument may be based on the covariant phase space formulation \cite{Wald:1990mme,Lee:1990nz,Iyer:1994ys}.}

In the examples we previously discussed, we could choose the slow degrees of freedom {\it freely} and then apply the order reduction scheme to fix data for the extra degrees of freedom in terms of data for the slow degrees of freedom. However, for theories with initial data constraints, the slow degrees of freedom ($\gamma_{\mu \nu}$, $K_{\mu \nu}$ in gravity) cannot be initialized completely freely.
Nevertheless, order reduction can be used on the spatial projection of the equations of motion $\gamma_\mu^\rho \gamma^\sigma_\nu E_{\rho\sigma}=0$ to fix the data for $(\pounds_n)^k K_{\mu \nu}$ (with $k\geq 1$) in terms of the data for $\gamma_{\mu \nu}$, $K_{\mu \nu}$. More precisely, one can derive expressions of the form
\begin{equation}
    (\pounds_n)^k K_{\mu \nu}=\kappa^{(k)}_{\mu \nu}(\gamma, R_\gamma, D R_\gamma, \ldots, K, DK, \ldots),
\end{equation}
where $D_\mu$ is the covariant derivative associated with the induced metric $\gamma$ and $R_\gamma$ is the Riemann curvature associated with $\gamma$. Then, substituting these expressions into $E^{\mu \nu} n_\nu=0$ gives constraint equations for $\gamma_{\mu \nu}$, $K_{\mu \nu}$. We shall refer to the resulting constraint equations as order-reduced constraints. Of course, the order-reduced constraints will contain higher than second-order spatial derivatives of $\gamma_{\mu \nu}$, $K_{\mu \nu}$. The character of these equations is unclear. We leave to future work the problem of developing methods for solving these equations. 

Even though finding exact solutions to the order-reduced constraint equations of general gravitational EFTs seems very challenging, there is a straightforward way to find approximate (perturbative) solutions that solve the constraints modulo an error term consistent with the expected accuracy of the EFT. Consider a gravitational EFT that is accurate up to an error term of order $\lambda^{\mathfrak{m}+1}$, where $\lambda$ is a fundamental length scale associated with UV physics. We can generate a perturbative solution,
\begin{equation}
    \gamma_{\mu \nu}=\sum\limits_{k=0}^{\mathfrak{m}}\lambda^k\gamma^{(k)}_{\mu \nu},\qquad \qquad K_{\mu \nu}=\sum\limits_{k=0}^{\mathfrak{m}}\lambda^k K^{(k)}_{\mu \nu}.
\end{equation}
to the order-reduced EFT constraints up to an error of order $\lambda^{\mathfrak{m}+1}$. When this data is evolved forward in time using the evolution equations, our failure to solve the constraints on the initial data surface will introduce a secularly growing error to the constraints at later times (which follows from the equations governing the evolution of constraint violations). This problem may be alleviated by using a trick well-known in numerical relativity: by adding suitable constraint-damping terms (that vanish on-shell) to the evolution equations, one can make sure that the constraint violations satisfy damped wave equations, guaranteeing the decay of constraint violations at late times (see e.g. \cite{Gundlach:2005eh,Weyhausen:2011cg,Alic:2011gg}). This may be a suitable solution for our purposes, since solving the non-linear EFT evolution equations (that are themselves accurate only up to order $\lambda^{\mathfrak{m}+1}$) is already a source of a secularly growing error of order\footnote{
This assumes only that the truncation error of the EFT equations is bounded in time. If the truncation error actually decays (e.g., due to dispersion) then the secular growth may be slower, or absent.} $t\,\lambda^\mathfrak{m}$ \cite{Reall:2021ebq,Figueras:2025wtx}. In other words, even if we had an exact solution to the initial value constraints, the best we can do (generically) is to solve the EFT equations of motion up to an error of order $t\,\lambda^{\mathfrak{m}+1}$.

We will now illustrate these ideas for Einstein-Gauss-Bonnet (EGB) gravity in $\mathcal{D}>4$ spacetime dimensions,
\begin{equation}
    S=\frac{1}{16\pi G}\int d^\mathcal{D}x\sqrt{-g}\left[R+\alpha\, \lambda^2\,\mathcal{G}+\mathcal{O}(\lambda^4)\right]\,, \label{eq:EGB_action}
\end{equation}
where $\alpha$ is a dimensionless coupling constant, $\lambda$ a fundamental length scale, and
\begin{equation}
    \mathcal{G}=R^2 -4\,R_{\mu \nu}\,R^{\mu \nu}+R_{\mu \nu \rho \sigma}\,R^{\mu \nu \rho \sigma}\,.
\end{equation}
The equation of motion of \eqref{eq:EGB_action} is second order so setting up initial data does not require any order reduction. Moreover, at weak coupling (i.e. within the regime of validity of EFT), this theory admits a well-posed initial value formulation (in a modified harmonic gauge) \cite{Kovacs:2020ywu} (see also \cite{AresteSalo:2023mmd}). Nevertheless, we shall instead use this theory to illustrate the FHK procedure \cite{Figueras:2024bba}, which leads to a theory with equations of motion that are not second order, but still admit a well-posed initial value problem. We start by performing a perturbative field redefinition,
\begin{equation}
    g_{\mu \nu}\to g_{\mu \nu}+\beta\, \lambda^2\, R_{\mu \nu}+\mathcal{O}(\lambda^4),
\end{equation}
which produces the regularized EGB action
\begin{equation}
    S=\frac{1}{16\pi G}\int d^\mathcal{D} x\sqrt{-g}\left[R+ \lambda^2\,\left(\alpha\,\mathcal{G}+\beta\, R^{\mu \nu} G_{\mu \nu}\right)+\mathcal{O}(\lambda^4)\right]\,. \label{eq:reg_EGB_action}
\end{equation}
The resulting equations of motion are given by
\begin{align}
    E^\mu{}_\nu\equiv&\,G^\mu{}_{\nu} + \alpha\,\lambda^2\,\delta_{\mu \rho_1 \rho_2 \rho_3 \rho_4}^{\nu \sigma_1 \sigma_2 \sigma_3 \sigma_4} R_{\sigma_1\sigma_2}{}^{\rho_1 \rho_2}R_{\sigma_3\sigma_4}{}^{\rho_3 \rho_4}
     \nonumber\\
     &-\beta\,\lambda^2\, \left(\Box G^\mu{}_\nu-2\,G^\rho{}_\sigma R^{\mu \sigma}{}_{\nu \rho}-\frac12 \delta^\mu_\nu \, G^\sigma{}_\rho R^\rho{}_\sigma \right)+\mathcal{O}(\lambda^4)=0.
    \label{eq:eomsEGB}
\end{align}
This implies $G_{\mu\nu} = \mathcal{O}(\lambda^2)$, so the quantity appearing in the round bracket is of order $\lambda^2$ \textit{on-shell}. Hence the second line above is $\mathcal{O}(\lambda^4)$ on shell. This implies that the equations of motion of \eqref{eq:reg_EGB_action} and \eqref{eq:EGB_action} agree within the expected accuracy of the EFT.


The regularized four-derivative EFT has fourth-order equations, whereas the original EGB theory ($\beta=0$) has second-order equations. This means that the regularized theory requires Cauchy data for the metric and its time derivatives up to third order. More precisely, let $\Sigma$ be a $d=\mathcal{D}-1$-dimensional Cauchy slice with unit normal $n$. Then, initial data for the four derivative EFT consists of $(\Sigma, \gamma_{\mu \nu}, K_{\mu \nu}, (\pounds_n) K_{\mu \nu}, (\pounds_n)^2 K_{\mu \nu})$.
As explained above, one can then fix initial data for $(\pounds_n) K_{\mu \nu}$, $(\pounds_n)^2 K_{\mu \nu}$ (fast degrees of freedom) in terms of $\gamma_{\mu \nu}$, $K_{\mu \nu}$ (slow degrees of freedom) by employing an order reduction strategy, though the order reduction procedure is quite simple in this case, as we shall see now.

To carry out the order reduction, we consider the spatial projections of the equations of motion, which may be written as
\begin{equation}\label{eq:EGB_proj}
\begin{split}
\gamma^\mu_{\mu_1}\gamma^{\nu_1}_\nu E^{\mu_1}{}_{\nu_1}&=\,\gamma^\mu_{\mu_1}\gamma^{\nu_1}_\nu\,\left(G^{\mu_1}{}_{\nu_1} + \alpha\,\lambda^2\,\delta_{\nu_1\rho_1\rho_2 \rho_3 \rho_4}^{\mu_1\sigma_1 \sigma_2 \sigma_3 \sigma_4} R_{\sigma_1\sigma_2}{}^{\rho_1 \rho_2}R_{\sigma_3\sigma_4}{}^{\rho_3 \rho_4}\right)+O(\lambda^4)\\
     &=A^{\mu}{}_\nu{}^{\sigma\rho}\pounds_n K_{\sigma\rho}+M^\mu{}_\nu+O(\lambda^4)=0,
\end{split}
\end{equation}
where in the first line we used the fact that the ``regularising'' terms in the equations of motion are $O(\lambda^4)$, and in the second line we simply separated terms with $\pounds_n K_{\sigma\rho}$ (i.e. with second time derivatives of the metric) from everything else. Note that the tensors $A^\mu{}_\nu{}^{\sigma\rho}$ and $M^\mu_\nu$ depend only on $\gamma$, $R_\gamma$, $K$ and $DK$, a fact that follows from the analysis of the symmetries of the principal symbol discussed in \cite{Papallo:2017qvl}. Assuming that the initial data surface is non-characteristic, the tensor $A_{\mu \nu}{}^{\sigma\rho}$ (regarded as a linear map on the space of symmetric tensors) is invertible and one can obtain initial data for $\pounds_n K_{\mu \nu}$ in terms of the slow degrees of freedom:
\begin{equation}\label{eq:LnK}
    \pounds_n K_{\mu \nu}=-\left(A^{-1}\right)_{\mu \nu}^{\sigma\rho}M_{\sigma\rho}+O(\lambda^4)
\end{equation}
The initial data for $(\pounds_n)^2 K_{\mu \nu}$ is then obtained by taking a Lie derivative of \eqref{eq:LnK} and replacing any occurrence of $\pounds_n K_{\mu \nu}$ with the RHS of \eqref{eq:LnK}.

To solve the initial value constraints $E_{\mu \nu}n^\nu=0$ for the slow degrees of freedom, one can follow the well-known conformal methods used for solving the Einstein constraints \cite{Choquet-Bruhat:2009xil}. A conformal-transverse-traceless (CTT) decomposition allows for a recasting of the Einstein constraints as a system of second order elliptic PDEs that admit a well-posed boundary value problem under certain conditions. In the case of EGB theory, a similar strategy is also applicable: at weak coupling the conformally formulated constraints are still second order elliptic PDEs and a continuity argument may be used to establish the existence of an exact solution to the EGB constraint equations in a neighbourhood of a solution of the Einstein constraints \cite{Kovacs:2021lgk}. (See also \cite{Brady:2023dgu,Aurrekoetxea:2025kmm,Nee:2024bur} for related methods to solve the constraint equations of Einstein-scalar-Gauss-Bonnet gravity numerically.) Alternatively, one could solve the conformally formulated constraints in the form of a series expansion in $\lambda$. In this case, a perturbative solution is obtained by solving the constraints order-by-order: at each order one would solve the Einstein constraints sourced by terms depending on the lower order solutions. For a more detailed example of this in ``dynamical'' Chern-Simons gravity see \cite{Okounkova:2018abo}.

\section{Conclusions}
\vspace{-0.2cm}

In this paper we have proposed a method for dealing with the unphysical degrees of freedom appearing in well-posed formulations of the equations of various EFTs, including those of viscous hydrodynamics and of gravity. The proposal is based on the well-known idea of order reduction. The novelty is that this is applied only at the level of initial data rather than to the equations of motion. This procedure breaks formal Lorentz (or general) covariance, but the effect of this breaking is comparable to the size of the effects caused by higher-order corrections to the EFT equations of motion, as long as one restricts to inertial frames (or coordinate systems) in which EFT is manifestly valid. 

In the context of relativistic viscous hydrodynamics (see Sec. \ref{bdnksection}), we believe that the present findings remove the final potential roadblock towards a physically reliable numerical implementation of BDNK theory, fixing a long-lasting (but often overlooked \cite{Pandya:2021ief,Pandya:2022pif,Pandya:2022sff,Clarisse:2025lli}) initialization ambiguity. When BDNK is treated not merely as a set of evolution equations \eqref{BDNKEOM}, but as a complete EFT Cauchy problem equipped with a consistent initialization prescription [namely~\eqref{theuglymonster}], then the non-hydrodynamic modes are effectively absent, and the viscous fluid carries the same number of physical degrees of freedom as the ideal fluid.

We emphasize that, even if a system initially lies within the regime of validity of the EFT, there is no guarantee that the subsequent evolution will remain so. In general, non-linear effects in both hydrodynamic and gravitational systems can drive the formation of singular structures, such as shocks or black holes, where the EFT description necessarily breaks down. When this occurs, the agreement between the well-posed equations of motion and their order-reduced counterpart is generically lost, signaling the failure of the underlying EFT assumptions. In this sense, the order-reduced equation may also be used as a diagnostic tool, providing a practical criterion to monitor the consistency of the EFT approximation during the evolution. Nevertheless, the EFT may still provide an accurate effective description in regions sufficiently far from these singularities. For instance, hydrodynamics remains applicable at distances of several mean free paths from a shock front, while gravitational EFTs are expected to be reliable outside black hole horizons. In such situations, the initialization procedure proposed here should be applied only within the spacetime regions where the EFT assumptions remain valid.

\vspace{-0.2cm}
\section*{Acknowledgments}
\vspace{-0.2cm}

LG is supported by a MERAC Foundation prize grant,  an Isaac Newton Trust Grant, and funding from the Cambridge Centre for Theoretical Cosmology. HSR is supported by STFC grant no. ST/X000664/1. ADK is supported by the STFC Consolidated Grant ST/X000931/1. ADK is grateful to Katy Clough, Pau Figueras and Shunhui Yao for helpful discussions.

\appendix

\vspace{-0.2cm}
\section{The extra degrees of freedom are always fast (non-linear argument)}\label{appAAA}
\vspace{-0.2cm}

Theorem \ref{theo1} does not admit a straightforward generalization to the nonlinear regime, since there is no well-defined notion of a dispersion relation. Nevertheless, one can still convince oneself that the claim remains valid through the following simple argument.

Assume, for simplicity, that the system is spatially homogeneous, so the derivatives $\partial_j$ can be discarded. Then, the field $\phi(t)$ obeys some ordinary differential equation, for example
\begin{equation}\label{odeeom}
F(\phi,\dot{\phi})-\lambda \dddot{\phi}=0 \, ,
\end{equation}
with $F$ some non-linear function.
For $\lambda=0$, the equation reduces to a first-order one, so we only need to specify the initial value of $\phi$. However, when $\lambda\neq 0$, the equation becomes of third order, and one must additionally specify the initial values of $\dot{\phi}$ and $\ddot{\phi}$. Now, let us isolate the higher-derivative term and evaluate the equation at $t=0$, obtaining
\begin{equation}
\dddot{\phi}(0) =\dfrac{1}{\lambda} F\big(\phi(0),\dot{\phi}(0)\big)\, .
\end{equation}
As $\lambda\rightarrow 0$, it becomes evident that, for a given $\phi(0)$, a generic choice of $\dot{\phi}(0)$ leads to a divergent $\dddot{\phi}(0)$. The only way to prevent this divergence is to require that $F\big(\phi(0),\dot{\phi}(0)\big)\approx 0$, thereby fixing $\dot{\phi}(0)$ as a function of $\phi(0)$.
Proceeding further, differentiating \eqref{odeeom} and evaluating at $t=0$ gives
\begin{equation}
\ddddot{\phi}(0) =\dfrac{1}{\lambda} \left[\dfrac{\partial F}{\partial\phi}\bigg|_{t=0} \dot{\phi}(0) +\dfrac{\partial F}{\partial\dot{\phi}}\bigg|_{t=0} \ddot{\phi}(0) \right]\, .
\end{equation}
Again, unless $\dot{\phi}(0)$ and $\ddot{\phi}(0)$ are chosen in a very specific way, the fourth derivative $\ddddot{\phi}(0)$ diverges for small $\lambda$. Hence, the two additional degrees of freedom associated with freely specifying $\dot{\phi}$ and $\ddot{\phi}$ correspond to states where $\dddot{\phi}$ and/or $\ddddot{\phi}$ become very large, supporting the interpretation of these degrees of freedom as ``fast'' modes.

\section{Reducing multiple orders (non-linear argument)}\label{AppBBB}

Let us demonstrate that an iterative procedure, similar to that employed in the proof of Theorem \ref{theo2}, can also be applied to reduce the time order of the equations in the non-linear regime. Similarly to appendix \ref{appAAA}, we will assume that the space derivatives can be discarded, and we will therefore consider the ordinary differential equation
\begin{equation}
\phi^{(\mathfrak{n})}=F(\phi,\dot{\phi},...,\phi^{(\mathfrak{n}-1)})+\lambda G(\phi,\dot{\phi},...,\phi^{(\mathfrak{n})},...,\phi^{(\mathfrak{m})})+\mathcal{O}(\lambda^2)\, ,
\end{equation}
which, truncated to zeroth order in $\lambda$, becomes
\begin{equation}
\phi^{(\mathfrak{n})}=F(\phi,\dot{\phi},...,\phi^{(\mathfrak{n}-1)})+\mathcal{O}(\lambda)\, .
\end{equation}
Then, the reduction procedure works a follows:
\begin{equation}
\begin{split}
\phi^{(\mathfrak{n})}={}& F\left(\phi,\dot{\phi},...,\phi^{(\mathfrak{n}-1)}\right)+\lambda G\left(\phi,\dot{\phi},...,\phi^{(\mathfrak{n})},...,\phi^{(\mathfrak{m})}\right)+\mathcal{O}(\lambda^2)\\
={}& F\left(\phi,\dot{\phi},...,\phi^{(\mathfrak{n}-1)}\right)+\lambda G\left(\phi,\dot{\phi},...,\phi^{(\mathfrak{n})},...,\phi^{(\mathfrak{m}-1)},\dfrac{d^{\mathfrak{m}-\mathfrak{n}}}{dt^{\mathfrak{m}-\mathfrak{n}}}F\right)+\mathcal{O}(\lambda^2) \\
={}& F\left(\phi,\dot{\phi},...,\phi^{(\mathfrak{n}-1)}\right)+\lambda G_1\left(\phi,\dot{\phi},...,\phi^{(\mathfrak{n})},...,\phi^{(\mathfrak{m}-1)}\right)+\mathcal{O}(\lambda^2)\\
={}& F\left(\phi,\dot{\phi},...,\phi^{(\mathfrak{n}-1)}\right)+\lambda G_1\left(\phi,\dot{\phi},...,\phi^{(\mathfrak{n})},...,\phi^{(\mathfrak{m}-2)},\dfrac{d^{\mathfrak{m}-\mathfrak{n}-1}}{dt^{\mathfrak{m}-\mathfrak{n}-1}}F\right)+\mathcal{O}(\lambda^2)\\
={}& F\left(\phi,\dot{\phi},...,\phi^{(\mathfrak{n}-1)}\right)+\lambda G_2\left(\phi,\dot{\phi},...,\phi^{(\mathfrak{n})},...,\phi^{(\mathfrak{m}-2)}\right)+\mathcal{O}(\lambda^2)\\
={}& ...=F\left(\phi,\dot{\phi},...,\phi^{(\mathfrak{n}-1)}\right)+\lambda G_{\mathfrak{m}-\mathfrak{n}+1}\left(\phi,\dot{\phi},...,\phi^{(\mathfrak{n}-1)}\right)+\mathcal{O}(\lambda^2)\, .\\
\end{split}
\end{equation}

\bibliography{Biblio}

\label{lastpage}
\end{document}